\documentclass{article}
\usepackage{fullpage}

\usepackage{amssymb}

\newcommand{\R}{\ensuremath{\mathbb{R}}}

\usepackage{amsmath}

\newcommand{\cbra}[1]{\left\{#1\right\}}

\usepackage{amsthm}
\usepackage[hidelinks]{hyperref}
\hypersetup{
    colorlinks=true,
    linkcolor=blue,
    filecolor=magenta,
    urlcolor=cyan,
    citecolor=magenta,
}
\usepackage{thmtools}
\usepackage{thm-restate}
\usepackage{cleveref}

\newtheorem{theorem}{Theorem}[section]
\newtheorem{lemma}[theorem]{Lemma}

\newtheorem{definition}[theorem]{Definition}
\newtheorem{corollary}[theorem]{Corollary}
\newtheorem{claim}[theorem]{Claim}
\newtheorem{fact}[theorem]{Fact}
\newtheorem{conjecture}[theorem]{Conjecture}

\newcommand{\norm}[1]{\ensuremath{\left\|#1\right\|}}
\newcommand{\zone}{\{0,1\}}
\newcommand{\adeg}{\widetilde{\mathrm{deg}}}
\newcommand{\PDT}{\mathsf{PDT}}
\newcommand{\DT}{\mathsf{DT}}
\newcommand{\OR}{\mathsf{OR}}
\newcommand{\XOR}{\mathsf{XOR}}
\newcommand{\SWk}{\mathsf{SW}_k}

\usepackage{tikz}

\usepackage{xcolor}
\usepackage[dvipsnames]{xcolor}

\usepackage{subfiles}

\title{Subcube Stifling}

\author{Arjan Cornelissen\thanks{Simons Institute for the Theory of Computing, University of California, Berkeley, United States of America {\tt ajcornelissen@outlook.com}}
\and
Nikhil S.~Mande\thanks{University of Liverpool, UK {\tt mande@liverpool.ac.uk}}
\and 
Nithish Raja\thanks{Eindhoven University of Technology, Netherlands {\tt n.r.raja@tue.nl}}
}

\date{}
\begin{document}
\maketitle

\begin{abstract}
We introduce the \emph{subcube stifling number}, a new combinatorial measure of total Boolean functions. This measure is the largest integer $k$ such that, for every set $S$ of at most $k$ input variables and every assignment $b \in \{0,1\}^S$, there is a fixing of the variables outside $S$ under which the resulting function on the free variables $S$ is the point indicator $\mathbb{I}[x_S=b]$. Equivalently, for every small set of coordinates, the function can isolate any prescribed point of the corresponding Boolean cube by suitably fixing all remaining coordinates. This measure is inspired by the stifling number of Chattopadhyay et al.~(ITCS'23); whereas their measure asks for restrictions realizing every constant function, ours asks for restrictions realizing every point indicator.
Our results are as follows.
\begin{itemize}
    \item We show that the subcube stifling number gives rise to an approximate-degree composition theorem. In particular, if a Boolean function $f$ satisfies $\adeg(f)=O(\sqrt{\mu(f)})$, then for every Boolean function $g$, approximate degree composes tightly:
    \[
        \adeg(f \circ g)=\Theta(\adeg(f)\adeg(g)).
    \]
    This motivates the study of the subcube stifling number, and in particular the search for functions whose approximate degree is $O(\sqrt{\mu(f)})$.
    \item We show using a standard probabilistic argument that a random Boolean function on $n$ input bits has subcube stifling number \(\Theta(\log n)\) with high probability.
    \item Chattopadhyay et al.~showed that the $\mathsf{Majority}$ function has linear stifling number, and no Boolean function has larger stifling number. In contrast, the subcube stifling number of $\mathsf{Majority}$ is easily seen to be 0. This raises the question whether there even exist Boolean functions with linear subcube stifling number.  
    
    We show that this is indeed the case. Our examples are obtained from indicators of linear codes over $\mathbb{F}_2$ whose minimum distance and dual distance are both linear.
    \item We prove that the functions arising from this linear-code construction do not have approximate degree $O(\sqrt{\mu(f)})$; in fact, they have approximate degree $\Omega(\mu(f))$.
\end{itemize}

The main question left open is whether there exists a Boolean function $f$ with $\adeg(f)=\Theta(\sqrt{\mu(f)})$. A positive answer would yield new instances of tight approximate-degree composition with $f$ as the outer function.
\end{abstract}

\newpage

\section{Introduction}

We define a combinatorial measure called the \emph{subcube stifling number} of a total Boolean function $f : \{0,1\}^n \to \{0,1\}$, which for simplicity we denote by $\mu(f)$ throughout this paper. Informally $\mu(f)$ is the largest value of $k$ such that for all sets $S$ of size at most $k$ and all choices of $b \in \zone^k$, there is a way to set the variables outside $S$ such that the restricted function equals the point function $\mathbb{I}[y = b]$.
Chattopadhyay et al.~\cite{Chattopadhyay_Mande_Sanyal_Sherif_2023} defined a related notion called the \emph{stifling number} of a Boolean function: it is the largest $k$ such that for every subset $S$ of at most $k$ variables and every $b \in \zone$, there exists a setting of the variables outside $S$ such that the restricted function equals the constant $b$. 
Our measure, the subcube stifling number, differs from the stifling number in that it requires the restricted function to behave like an indicator function instead of a constant. We refer the reader to Section~\ref{sec:prelims} for formal definitions of subcube stifling number and stifling number.

Chattopadhyay et al.~\cite{Chattopadhyay_Mande_Sanyal_Sherif_2023} showed a lifting theorem involving stifling number: if the stifling number of a gadget $g$ is large, then the decision tree complexity of $f$ lifts to parity decision tree complexity of $f \circ g$.\footnote{They showed a stronger version of this statement, but we only state the depth-to-PDT-depth lifting result for simplicity.}

\begin{theorem}[{\cite[Theorem 4]{Chattopadhyay_Mande_Sanyal_Sherif_2023}}]\label{thm: pdt depth lifting}
Let $g:\{0,1\}^m \to \{0,1\}$ be a $k$-stifled function, and let $f \subseteq \zone^n \times \mathcal{R}$ be a relation. Then $\PDT(f \circ g) \geq \DT(f) \cdot k$.
\end{theorem}

Interestingly, our variant $\mu$ also yields a lifting theorem with a very different flavor. In order to motivate our lifting theorem, we first define approximate degree.

\begin{definition}\label{def:adeg}
    Let $f:\{0,1\}^n\rightarrow\{0,1\}$ be a Boolean function. Then, a polynomial $p$ is an $\varepsilon$-approximating polynomial of $f$ if it satisfies the following.
    \begin{equation*}
        |f(x) - p(x)| \le \varepsilon, \qquad \forall x\in\{0,1\}^n.
    \end{equation*}
    The $\varepsilon$-approximate degree of $f$, denoted by $\adeg_\varepsilon(f)$, is defined as $\min_p\deg(p)$ such that the polynomial $p$ satisfies the above condition. When $\varepsilon = 1/3$, we denote approximate degree by $\adeg(f)$.
\end{definition}

Investigating how complexity measures behave under composition has been a widely studied area of research \cite{Saks_Wigderson_1986, Nisan_Szegedy_1994, Jukna_Razborov_Savicky_Wenger_1999, Ambainis_2005, Tal_2013, Gilmer_Saks_Srinivasan_2016}.
Composition of approximate degree is no exception. Sherstov~\cite{Sherstov_2013} showed that for all Boolean $f, g$ the upper bound $\adeg(f \circ g) = O(\adeg(f) \cdot \adeg(g))$ holds. Showing a matching lower bound remains a significant open problem. Nevertheless, there has been some progress on this front. Ben-David et al.~\cite{Ben-David_Bouland_Garg_Kothari_2018} showed that $\adeg(f \circ g) = \Omega(\adeg(f) \cdot \adeg(g))$ for all $g$, whenever $f$ is a symmetric function. Chakraborty et al.~\cite{Chakraborty_Kayal_Mittal_Paraashar_Saurabh_2024} showed that the same lower bound holds for all recursively defined $f$.

It is well known that the approximate degree of the $n$-variate $\mathsf{AND}$ function is $\Theta(\sqrt{n})$~\cite{Nisan_Szegedy_1994}.
It follows fairly easily that for all Boolean functions $f : \zone^n \to \zone$, we have $\adeg(f) = \Omega(\sqrt{\mu(f)})$. This is because $f$ embeds an $\mathsf{AND}$ function on $\mu(f)$ many variables (in fact it embeds several indicators on $\mu(f)$ many variables).
We show that if this naive lower bound is tight, then approximate degree composition holds for $f$. That is,

\begin{restatable}{theorem}{adeglifting}\label{thm: subcube stifling adeg lifting}
    Let $f : \zone^n \to \zone$ be a Boolean function. Then, for all Boolean functions $g$ we have
    \[
    \adeg(f \circ g) = \Omega(\sqrt{\mu(f)} \adeg(g)).
    \]
\end{restatable}
In particular, this implies that if a Boolean function $f$ satisfies $\adeg(f) = \Theta(\sqrt{\mu(f)})$, then approximate degree composition holds for all inner $g$, i.e., $\adeg(f \circ g) = \Theta(\adeg(f)\adeg(g))$ (the upper bound follows from~\cite{Sherstov_2013}).
This result is similar in flavor to one of the main results of~\cite{Chakraborty_Kayal_Mittal_Paraashar_Sanyal_Saurabh_2023}, who showed that $\adeg(f \circ g) = \widetilde\Omega(\sqrt{\mathsf{bs}(f)} \adeg(g))$, where $\mathsf{bs}(\cdot)$ denotes block sensitivity. We refer the reader to \Cref{sec:comparison} for a comparison of our result with theirs.

Our lower bound proof is a reasonably straightforward adaptation of the proof of~\cite[Theorem~1]{Ben-David_Bouland_Garg_Kothari_2018}, who showed that $\adeg(\OR \circ g) = \Theta(\adeg(\OR) \cdot \adeg(g))$ for all Boolean functions $g$. As a crucial component of their proof, they use a result of Belovs~\cite{Belovs_2015} which shows that the parity of any $n$-bit string can be computed quantumly using $\Theta(\sqrt{n})$ given query access to ORs of arbitrary subsets of variables. In contrast, we use a recent result~\cite{Cornelissen_Mande_Patro_Raja_Sanyal_2025} which shows that the parity of any $n$-bit string can be computed quantumly using $\Theta(n/\sqrt{k})$ queries, given query access to $\mathbb{I}[x_S = b]$ for all $S \subseteq [n]$ with $|S| \leq k$ and all $b \in \zone^S$. The rest of our proof follows exactly the same framework as that of~\cite[Theorem~1]{Ben-David_Bouland_Garg_Kothari_2018}.

We find it very interesting that stifling number and our subcube stifling number look fairly similar, yet give rise to wildly different lifting theorems. On the one hand, stifling number plays a role in lifting decision-tree complexity as a measure of the inner gadget (\Cref{thm: pdt depth lifting}). On the other hand, subcube stifling number plays a role in an approximate degree composition result as a measure of the outer function (\Cref{thm: subcube stifling adeg lifting}).
For these reasons, we find it worthwhile to perform a systematic study of subcube stifling number.
Our results are summarized in the next subsection.

\subsection{Our results}

We make the following contributions in this work.

\begin{itemize}
    \item We show Theorem~\ref{thm: subcube stifling adeg lifting} using a reasonably simple adaptation of the proof of~\cite[Theorem~1]{Ben-David_Bouland_Garg_Kothari_2018}. This motivates the study of subcube stifling number. By noting that $f$ must embed an $\mathsf{AND}$ on $\mu(f)$ variables, we observe the following relationships between subcube stifling number and standard complexity measures of Boolean functions such as sensitivity and degree (see \Cref{claim:bounds_due_to_OR}): $\mathsf{s}(f) = \Omega(\mu(f))$, $\deg(f) = \Omega(\mu(f))$ and $\lambda(f) = \Omega(\sqrt{\mu(f)})$, see~\Cref{fig:complexity-measures}. Moreover, we exhibit functions witnessing tightness of all of these bounds (see last bullet below).

    \begin{figure}[!ht]
        \centering
        \begin{tikzpicture}[measure/.style = {blue, draw, rounded corners = .3em}]
            \node[measure] (lambda) at (0,0) {$\lambda$};
            \node[measure] (adeg) at (1,1) {$\widetilde{\deg}$};
            \node[measure] (s) at (-1,1) {$\mathsf{s}$};
            \node[measure] (bs) at (-2,2) {$\mathsf{bs}$};
            \node[measure] (Q) at (0,2) {$\mathsf{Q}$};
            \node[measure] (deg) at (2,2) {$\deg$};

            \draw (lambda) to (s) to (bs);
            \draw (lambda) to (adeg) to (Q);
            \draw (adeg) to (deg);
            \draw (lambda) to (Q);

            \draw (lambda) to (0,-.5) node[below] {$\sqrt{\mu}$};
            \draw (s) to (-1.5,.5) node[below left] {$\mu$};
            \draw (deg) to (2.5,1.5) node[below right] {$\mu$};

            \draw[red, dashed] plot [smooth] coordinates {(-.5,-.5) (-.5,2.5) (0.25,3) (2,1)};
        \end{tikzpicture}
        \caption{Overview of the complexity measures of Boolean functions considered here. We refer the reader to~\cite{Buhrman_de_Wolf_2002} and \Cref{def:spectral_sensitivity} for definitions of these measures. If two measures are connected, this means the upper one is always asymptotically bigger than the lower one for all Boolean functions. All measures above the red dashed line are at least $\Omega(\mu(\cdot))$, whereas all measures below are only known to be at least $\Omega(\sqrt{\mu(\cdot)})$.}
        \label{fig:complexity-measures}
    \end{figure}
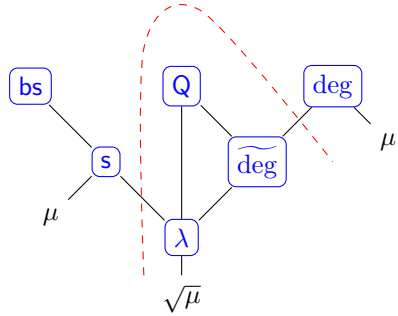
    
    \item We observe in \Cref{claim:subcube_stifling_of_random_func} using standard probabilistic arguments that a random Boolean function $f : \zone^n \to \zone$ has subcube stifling number $\Theta(\log n)$ with high probability. Since all the aforementioned measures are linear for random Boolean functions with high probability, this rules out the possibility of $\mu$ upper bounding any of these measures in general.
    
    \item It is easy to show that the $\mathsf{Majority}$ function has stifling number linear in the number of inputs, and that no Boolean function can achieve a larger stifling number~\cite[Claim 26]{Chattopadhyay_Mande_Sanyal_Sherif_2023}. In contrast, $\mathsf{Majority}$ has subcube stifling number zero. More generally, it is not a priori clear whether \emph{any} Boolean function attains linear subcube stifling number. We show in \Cref{cor:mulinearlb} that such functions do exist: they can be realized as indicator functions of linear codes over $\mathbb{F}_2$ (viewed as subsets of $\zone^n$) whose minimum distance and dual distance are both linear in $n$.
    We prove in \Cref{thm:subcube-stifling_upper_bound_for_codes} that the subcube stifling number for such functions satisfy $\mu(f) = \min\{d,d^{\perp}\} - 1$, where $d$ and $d^{\perp}$ are the distance and dual distance of the code, respectively. We then use a known result on random linear codes (see \Cref{thm: dim random}) to show that they, with high probability, give rise to Boolean functions with linear subcube stifling number.
    
    \item In view of Theorem~\ref{thm: subcube stifling adeg lifting}, it is a clear goal to identify explicit families of Boolean functions $f$ with $\adeg(f) = \Theta(\sqrt{\mu(f)})$ in order to exhibit new functions for which approximate degree composition holds. Towards this, it is natural to ask if there even exists an $f$ with $\adeg(f) = \Theta(\sqrt{\mu(f)})$. We show in \Cref{cor:mu_lower_bounds_approx_deg} that any function $f$ obtained using the construction in the previous bullet (i.e., the indicator function of a subspace) does not satisfy this property: for such functions we show $\adeg(f) \geq d^{\perp} = \Omega(\mu(f))$.
    We show in \Cref{cor:mutightness} that when $f$ is the indicator function of a linear code obtained in the previous bullet, we have $\mathsf{s}(f) = O(\mu(f))$, $\deg(f) = O(\mu(f))$ and $\lambda(f) = O(\sqrt{\mu(f)})$, showing tightness of the bounds mentioned in the first bullet above.

\end{itemize}
The main question that remains open from our work is whether there exists a Boolean function $f$ such that $\adeg(f) = \Theta(\sqrt{\mu(f)})$. If such functions exist, then Theorem~\ref{thm: subcube stifling adeg lifting} would imply that approximate degree composition holds true with these functions as the outer function. 
We conjecture that no such function exists.
\begin{conjecture}\label{conj}
For all Boolean functions $f : \zone^n \to \zone$, $\adeg(f) = \Omega(\mu(f))$.
\end{conjecture}
In particular, the last bound in the last bullet above ($\lambda(f) = O(\sqrt{\mu(f)})$) also rules out a natural approach to resolving \Cref{conj} by replacing $\adeg$ by $\lambda$ in the conjecture (recall from \Cref{fig:complexity-measures} that $\adeg(f) \geq \lambda(f)$ for all $f$).
Towards proving \Cref{conj}, it might be beneficial to consider the following weaker statement, which we also conjecture to hold.

\begin{conjecture}
    For all Boolean functions $f : \zone^n \to \zone$, $\mathsf{Q}(f) = \Omega(\mu(f))$.
\end{conjecture}

\subsection{Comparison with prior work}\label{sec:comparison}

Chakraborty et al.~\cite{Chakraborty_Kayal_Mittal_Paraashar_Sanyal_Saurabh_2023} proved several results on approximate degree composition and related questions. The result most relevant to us is their composition lower bound
\[
    \adeg(f \circ g) = \widetilde{\Omega}\!\left(\sqrt{\mathsf{bs}(f)}\,\adeg(g)\right)
\]
for every inner function $g$, where $\mathsf{bs}(\cdot)$ denotes block sensitivity. Since $\mathsf{bs}(f) \geq \mu(f)$ for all $f$ (see \Cref{claim:bounds_due_to_OR}), their result subsumes our \Cref{thm: subcube stifling adeg lifting}, modulo logarithmic factors. Indeed, prior to their work, no result of the form $\adeg(f \circ g) = \Omega(M(f) \cdot \adeg(g))$ was known for some non-trivial complexity measure $M(\cdot)$~\cite[Section~2.2]{Chakraborty_Kayal_Mittal_Paraashar_Sanyal_Saurabh_2023}. While their work settled the question up to logarithmic losses, it left open whether such losses are avoidable for some measure $M(\cdot)$. Our result gives the first way of eliminating this loss via the subcube stifling number.

Their main technical ingredient behind the above theorem of~\cite{Chakraborty_Kayal_Mittal_Paraashar_Sanyal_Saurabh_2023} is a lower bound of the form
\[
    \adeg(\mathsf{PrOR}\circ(h_1,\ldots,h_n))
    =
    \widetilde{\Omega}\!\left(\sqrt{n}\min_{i\in[n]}\adeg(h_i)\right),
\]
where $\mathsf{PrOR}$ denotes the Promise-$\mathsf{OR}$ function which is promised to receive as input a string of Hamming weight at most 1, and behaves exactly like $\mathsf{OR}$ on these inputs.  
Their proof of this closely follows the strategy of Ben-David et al.~\cite[Theorem~16, arXiv version]{Ben-David_Bouland_Garg_Kothari_2018}, who established the corresponding statement when all the functions $h_i$ are identical, with the main additional work being a more delicate analysis needed to handle non-identical inner functions. Once this ingredient is in place, the remainder of their composition argument follows the proof of~\cite[Theorem~16, arXiv version]{Ben-David_Bouland_Garg_Kothari_2018} without further substantive changes. A central ingredient in the proof of~\cite[Theorem~16, arXiv version]{Ben-David_Bouland_Garg_Kothari_2018} is a quantum algorithm for computing the parity of an input string in a variant of the combinatorial group testing model~\cite[Theorem~19, arXiv version]{Ben-David_Bouland_Garg_Kothari_2018}, generalizing an algorithm of Belovs~\cite[Theorem~3.1]{Belovs_2015}. In contrast to the approach of~\cite{Chakraborty_Kayal_Mittal_Paraashar_Sanyal_Saurabh_2023}, our contribution is not a refinement of the analysis around this ingredient, but a replacement of the ingredient itself: we use a new quantum algorithm for computing parity with special types of queries~\cite{Cornelissen_Mande_Patro_Raja_Sanyal_2025}, rather than the group-testing parity algorithm at the heart of the Ben-David et al.~framework.

Regardless of applications to approximate degree composition, we feel that subcube stifling number is an independently interesting measure to study in its own right.

\section{Preliminaries}\label{sec:prelims}
All logarithms are taken base 2, unless mentioned otherwise. For a positive integer $n$, we use the notation $[n]$ to denote the set $\cbra{1, 2, \dots, n}$. For any subset $S \subseteq [n]$, we write its complement $\bar{S} = [n] \setminus S$. We use $\exp(x)$ to denote $e^{x}$, where $e$ is Euler's number and is $\approx 2.718$. We abuse notation and use $p(n) = O(q(n))$ ($= \Omega(q(n)), = \Theta(q(n))$ to mean $f(n) \in O(q(n))$ ($\in \Omega(q(n)), \in \Theta(q(n))$, respectively). We use the notation $\mathsf{supp}(f)$ to denote the set of elements in the domain of the function that map to a $1$ output. We denote the hamming weight of a binary string $x$ by $\mathsf{wt}_H(x)$ and we use $x_T$ to refer to the bits in $x$ restricted to the indices in $T\subseteq [n]$.

\begin{definition}\label{def:basic_boolean_functions}
    We define some standard Boolean functions below.
    \begin{itemize}
        \item{\makebox[4cm][l]{$\mathsf{OR}:\{0,1\}^n\rightarrow\{0,1\}: $} $ \mathsf{OR}(x) = 1 \iff \mathsf{wt}_H(x) > 0$,}
        \item{\makebox[4cm][l]{$\mathsf{AND}:\{0,1\}^n\rightarrow\{0,1\}: $} $ \mathsf{AND}(x) = 1 \iff \mathsf{wt}_H(x) = n$,}
        \item{\makebox[4cm][l]{$\mathsf{Maj}:\{0,1\}^n\rightarrow\{0,1\}: $} $ \mathsf{Maj}(x) = 1 \iff \mathsf{wt}_H(x) > n/2$,}
        \item{\makebox[4cm][l]{$\mathsf{XOR}:\{0,1\}^n\rightarrow\{0,1\}: $} $ \mathsf{XOR}(x) = 1 \iff \mathsf{wt}_H(x)$ is odd.}
    \end{itemize}
\end{definition}

Next we define composition of Boolean functions.

\begin{definition}[Composition of Boolean functions]\label{def:boolean_function_compostion}
    Let $f:\{0,1\}^n\rightarrow\{0,1\}$ and $g:\{0,1\}^m\rightarrow\{0,1\}$ be Boolean functions. The function $h:\{0,1\}^{mn}\rightarrow\{0,1\}$ is said to be a composition of $f$ and $g$, denoted by $f\circ g$, if it is defined as follows.
    \begin{equation*}
        h(x_{11}, \ldots, x_{1m}, \ldots, x_{n1}, \ldots, x_{nm}) = f(g(x_{11}, \ldots x_{1m}), \ldots, g(x_{n1}, \ldots, x_{nm})).
    \end{equation*}
\end{definition}

We define the \emph{subcube stifling number} of $f$ to be the largest integer $k$ such that, for every choice of at most $k$ input variables $S$ and every assignment $b$ to those variables, one can fix the remaining variables so that the restriction of $f$ on $\zone^S$ is the indicator function of $x_S = b$.

Formally, for $f : \zone^n \to \zone, S \subseteq [n], c \in \zone^{\bar{S}}$, we use $f_{S \vert c}$ to refer to the Boolean function from $\zone^S \to \zone$ obtained by taking $f$ and restricting the bits in $\bar{S}$ to equal to the string $c$. This allows us to introduce the subcube stifling number formally.

\begin{definition}[Subcube stifling number]\label{def:ssnum}
    Let $f : \zone^n \to \zone$ be a Boolean function. The subcube stifling number of $f$, denoted $\mu(f)$, is defined as follows.
\[
\mu(f) \;=\; \max \left\{ k \in \mathbb{N} \;:\;
\begin{array}{l}
\forall S \subseteq [n] \text{ with } |S|\le k, \\
\forall b \in \{0,1\}^S, \\
\exists\, c \in \{0,1\}^{[n]\setminus S}
\text{ such that } \\
f_{S \vert c}
= \mathbb{I}[x_S = b]
\end{array}
\right\}.
\]
\end{definition}

Recall the related notion of a Boolean function being stifled~\cite{Chattopadhyay_Mande_Sanyal_Sherif_2023}.
\begin{definition}\label{def: stifled}
Let $g : \zone^m \to \zone$ be a Boolean function and $k \in [m]$. We say that $g$ \emph{is $k$-stifled} if the following holds:
\begin{align*}
&\forall S \subseteq [m]~\textnormal{with}~|S| \leq k~\textnormal{and}~\forall b \in \zone,\\ &\exists~z \in \zone^{[m] \setminus S}~\textnormal{such that for all}~x \in \zone^m ~\textnormal{with}~ x_{[m] \setminus S} = z, g(x) = b.
\end{align*}
\end{definition}

\subsection{Approximate degree}

The $\varepsilon$-approximate degree of a Boolean function is defined as the minimum degree of a real polynomial $p$ such that $|p(x) - f(x)| \leq \varepsilon$ for all $x$ in the domain of $f$. The approximate degree of a Boolean function can be captured by a particular linear program.
Strong duality then implies the following (for a proof, see, for example,~\cite[Theorem~6.1]{Bun_Thaler_2022}).

\begin{lemma}[Dual witness]\label{lem:dual}
    Let $f:\{0,1\}^n\rightarrow\{0,1\}$ be a Boolean function. Then, $\adeg_\varepsilon(f) = \Omega(d)$ if there exists a dual witness $\psi:\{0,1\}^n\rightarrow\mathbb{R}$ that satisfies the following conditions:
    \begin{align*}
        & \sum_{x}|\psi(x)| = 1, \tag* {\textbf{L1 norm}}\\
        & \sum_{x}\psi(x)f(x) > \varepsilon, \tag* {\textbf{Correlation}}\\
        & \sum_{x}\psi(x)\chi_{T}(x) = 0, \qquad \forall T\subseteq[n], |T| \le d. \tag* {\textbf{Pure High Degree}}
    \end{align*}
\end{lemma}

\begin{theorem}[\cite{Nisan_Szegedy_1994}]\label{thm:adeg and}
    Let $\mathsf{AND}:\{0,1\}^n\rightarrow \{0,1\}$ be a Boolean function such that $\mathsf{AND}(x) = 1$ if and only if $x = 1^n$. Then,
    \begin{equation*}
        \adeg(\mathsf{AND}) = \Theta(\sqrt{n}).
    \end{equation*}
\end{theorem}
We now observe some simple lower bounds on measures in terms of $\mu(f)$.
\begin{claim}\label{claim:bounds_due_to_OR}
    Let $f:\{0,1\}^n\rightarrow\{0,1\}$ be a Boolean function. Then,
    \begin{enumerate}
        \item $\adeg(f) = \Omega(\sqrt{\mu(f)})$,
        \item $\deg(f) = \Omega(\mu(f))$,
        \item $\mathsf{s}(f) = \Omega(\mu(f))$.
    \end{enumerate}
\end{claim}

\begin{proof}

    Let $S\subseteq [n]$, $|S| = \mu(f)$ and $b = 1^S$. The definition of subcube stifling number implies that there exists $c\in\{0,1\}^{\bar{S}}$ such that
    \begin{equation}
        f_{S\vert c} = \mathbb{I}[y = 1^k] = \mathsf{AND}_{\mu(f)}.
    \end{equation}

    \begin{enumerate}
        \item Let $p$ be an $\varepsilon$-approximating polynomial of $f$ with $\widetilde{\deg}_\varepsilon(f) = \deg(p)$. Then, the polynomial $p_{S\vert c}$ (obtained by fixing indices in $\bar{S}$ to $c$) $\varepsilon$-approximates $\mathsf{AND}_{\mu(f)}$. By Theorem~\ref{thm:adeg and},
        \begin{equation*}
            \adeg(f) = \deg(p) \geq \deg(p_{S\vert c}) = \Omega(\sqrt{|S|}) = \Omega(\sqrt{\mu(f)}).
        \end{equation*}

        \item Let $p$ be the representing polynomial of $f$. Then, $p_{S\vert c}$ is the representing polynomial of $\mathsf{AND}_{\mu(f)}$. Using arguments analogous to above (and the well-known fact that $\deg(\mathsf{AND}_k) = k$), we get $\deg(f) = \Omega(\mu(f))$.

        \item We know that $f_{S\vert c} = \mathsf{AND}_{\mu(f)} \implies \mathsf{s}(f_{S\vert c}) = \Omega(\mu(f))$. Therefore, the function $f$ has sensitivity $\Omega(\mu(f))$ on input $(1^S,c)$. Since sensitivity of a Boolean function is defined as maximum over all inputs, we get $\mathsf{s}(f) = \Omega(\mu(f))$.\qedhere
    \end{enumerate}
\end{proof}

\begin{theorem}[{\cite[Theorem~4.8]{Beals_Buhrman_Cleve_Mosca_de_Wolf_2001}}]\label{thm:polynomial_method}
    Let $f:\{0,1\}^n\rightarrow\{0,1\}$ be a Boolean function and let $\mathsf{Q}(f)$ denote the bounded-error quantum query complexity of function $f$. Then,
    \begin{equation*}
        \mathsf{Q}(f) = \Omega(\adeg(f)).
    \end{equation*}
\end{theorem}

\begin{theorem}[{\cite[Theorem~1.1]{Sherstov_2013}}]\label{thm:robust_polynomials}
    Let $p : \{0,1\}^n \to [-1/3,4/3]$ be a polynomial with degree $d$. Then, there exists a polynomial $p_{\mathrm{robust}}$ of degree $O(d)$ that satisfies the following.
    \begin{equation*}
        |p(x) - p_{\mathrm{robust}}(x + \Delta)| \le \frac{1}{3}, \qquad \text{for all } x \in \{0,1\}^n \text{ and } \Delta \in \left[-\frac{1}{3},\frac{1}{3}\right]^n
    \end{equation*}
\end{theorem}

\begin{theorem}[{\cite[Theorem~6.6]{Sherstov_2012}}]\label{thm:XOR_composition}
    Let $g:\{0,1\}^n\rightarrow\{0,1\}$ be any Boolean function. Then,
    \begin{equation*}
        \adeg(\mathsf{XOR}\circ g)\in\Omega(n\cdot\widetilde{\deg}(g)).
    \end{equation*}
\end{theorem}

\begin{definition}[Search with bounded wildcards]\label{def:search_with_bounded_wildcards}
    Let $1 \leq k \leq n$ be positive integers. Let $X := \{(S,b): S\subseteq[n],\ |S|\le k,\ b\in\{0,1\}^S\}$. For every $x\in\{0,1\}^n$, define $z\in\{0,1\}^X$\footnote{We suppress the dependence of $z$ on $x$ for ease of readability.} by $z_{S,b} := \mathbb{I}[x_S=b]$. Let $D := \{z : x\in\{0,1\}^n\}\subseteq\{0,1\}^X$. Then, the ``Search with $k$-bounded wildcards'' (partial) Boolean function (denoted by $\SWk$) is defined as follows.
    \begin{align*}
        \SWk:D\to\{0,1\} \quad \text{ and } \quad \SWk(z) = \mathsf{XOR}(x), \quad \forall x\in\{0,1\}^n.
    \end{align*}
\end{definition}

\begin{theorem}[{\cite[Restatement of Theorem~4.2]{Cornelissen_Mande_Patro_Raja_Sanyal_2025}}]\label{thm:cmprs}
The quantum query complexity of the ``Search with $k$-bounded wildcards'' Boolean function is $\Theta(n/\sqrt{k})$ i.e.,
\[
\mathsf{Q}(\SWk)=\Theta\left(\frac{n}{\sqrt{k}}\right).
\]
\end{theorem}
Note that the original theorem statement in \cite{Cornelissen_Mande_Patro_Raja_Sanyal_2025} considers the problem of computing parity of Boolean strings of length $n$ using bounded subcube queries. In \Cref{thm:cmprs}, we have simply restated this as a problem of computing a partial function (on a much larger input string, where each bit corresponds to the answer of a query) using standard index queries.

\begin{definition}[Spectral sensitivity]\label{def:spectral_sensitivity}
    Let $f:\{0,1\}^n\rightarrow\{0,1\}$ be a Boolean function. Let $G = (V,E)$ be the \emph{sensitivity graph} corresponding to Boolean function $f$ where $V = \{0,1\}^n$ and $E = \{\{x,y\} \mid \mathsf{wt}_H(x\oplus y) = 1, f(x) \neq f(y)\}$. We denote the adjacency matrix of graph $G$ using $A_f$. Then,
    \begin{equation*}
        \lambda(f) := \norm{A_f}.
    \end{equation*}
\end{definition}

\begin{theorem}[{\cite[Theorem~17]{Aaronson_Ben-David_Kothari_Rao_Tal_2021}}]\label{thm:spectral_sensitivity_lower_bound_adeg}
    Let $f:\{0,1\}^n\rightarrow\{0,1\}$ be a Boolean function. Then,
    \begin{equation*}
        \adeg(f) = \Omega(\lambda(f)).
    \end{equation*}
\end{theorem}

\subsection{Binary linear codes}

We identify $\zone^n$ with the vector space $\mathbb{F}_2^n$, where addition is coordinate-wise modulo $2$.
A set $S \subseteq \mathbb{F}_2^n$ is a \emph{linear subspace} if: $0 \in S$, and for all $x,y \in S$, we have $x + y \in S$.
We use the term \emph{linear code} to refer to a linear subspace of $\mathbb{F}_2^n$, we denote it by $S\le \mathbb{F}_2^n$.

\begin{definition}[Dual code]
Let $S \subseteq \mathbb{F}_2^n$. The \emph{dual code} of $S$ is
\[
S^\perp := \cbra{y \in \mathbb{F}_2^n : \langle x,y\rangle = 0 \ \text{for all } x \in S},
\]
where the inner product is $\langle x,y\rangle := \sum_{i=1}^n x_i y_i \pmod 2$.
\end{definition}

\begin{definition}[Distance]
The \emph{distance} of a linear code $S$ is
\[
d(S) := \min_{x \in S \setminus \cbra{0}} \mathsf{wt}_H(x).
\]
The dual distance of code $S$ refers to the distance of the dual code $S^\perp$.
\end{definition}

\begin{theorem}[{\cite[Theorem~18]{Breuckmann_Eberhardt_2021}}]\label{thm: dim random}
For every sufficiently large $n$ and all constants $\delta \in (0, 0.11)$, there exists a linear code $S \le \mathbb{F}_2^n$ such that 
\[
d(S)\ge \delta n
\qquad\text{and}\qquad
d(S^\perp)\ge \delta n.
\]
\end{theorem}
We include a first-principles proof in the appendix.

\section{Small approximate degree yields approximate degree composition}

In this section, we adapt the argument of \cite[Theorem~1]{Ben-David_Bouland_Garg_Kothari_2018} to the setting where we have substring queries. We prove \Cref{thm: subcube stifling adeg lifting}, which we restate first for convenience.

\adeglifting*

\begin{proof}
    Recall from \Cref{def:search_with_bounded_wildcards} that the partial Boolean function $\SWk$ is defined as follows.
    \begin{align*}
        \SWk(z) & = \XOR(x_1,\ldots,x_n),\\
        \text{where} \quad z_{S,b} & = \mathbb{I}[x_S = b], \forall S\subseteq [n], |S|\le k, b\in\{0,1\}^{S}.
    \end{align*}
    Due to \Cref{thm:cmprs}, we know that there exists a quantum query algorithm that computes $\SWk$ with probability at least $2/3$, making $O(n/\sqrt{k})$ queries. Combining this with \Cref{thm:polynomial_method} shows that there exists a $O(n/\sqrt{k})$-degree polynomial $q$ in the variables $\cbra{z_{S, b} : S \subseteq [n], |S| \leq k, b \in \zone^S}$ that $1/3$-approximates $\XOR(x)$. That is, 
    for all $x \in \zone^n$, 
    \begin{equation*}
        \left| q(\ldots,z_{S,b},\ldots) - \XOR(x) \right| \le \frac{1}{3} \qquad \text{when } z_{S, b} = \mathbb{I}[x_S = b]~\forall S, b.
    \end{equation*}

    \Cref{thm:robust_polynomials} then implies the existence of a ``robust" polynomial $q_{\text{robust}}$ of degree $O(n/\sqrt{k})$ that $1/3$-approximates $\XOR(x)$:
    \begin{equation}\label{eq:robust}
        \left| q_{\text{robust}}(\ldots,z_{S,b}',\ldots) - \XOR(x) \right| \le \frac{1}{3}, \qquad \forall x\in\{0,1\}^n, \forall z_{S,b}' \in 
        \begin{cases}
            \left[\frac{2}{3},\frac{4}{3}\right], & \text{if } x_S = b\\
            \left[-\frac{1}{3},\frac{1}{3}\right], & \text{otherwise}
        \end{cases}.
    \end{equation}

    Now we turn to the functions $f$ and $g$. Say $\mu(f) = k$ and $g : \{0,1\}^m\rightarrow\{0,1\}$ is non-constant (otherwise the theorem is trivially true). Let $h:= f\circ g$, $T := \adeg(h)$ and $p:\{0,1\}^{mn}\rightarrow\mathbb{R}$ be a $T$-degree $1/3$-approximating polynomial of $h$. Let $S \subseteq [n]$ with $|S| \leq k$, and $b \in \{0,1\}^S$. Then, due to the subcube stifling property of $f$, there exists a bit string $c \in \{0,1\}^{\bar{S}}$, such that $f_{S \vert c} = \mathbb{I}[x_S = b]$. Since $g$ is non-constant, we can find inputs $y_j$, for all $j \in \bar{S}$, such that $g(y_j) = c_j$. 
    
    Now, by fixing the variables corresponding to the $y_j$'s in $p$, we obtain the polynomial $p_{S,b}$ that $1/3$-approximates $\mathbb{I}[x_S = b] \circ g$, with $\deg(p_{S,b}) \leq T$. That is,
    \begin{align}\label{eqn:restr}
        \nonumber& \text{for all } S \subseteq [n] : |S| \leq k,\forall b \in \zone^S, \exists p_{S, b} : \zone^{mn} \rightarrow \mathbb{R}, \text{ such that } \deg(p_{S, b}) \leq T \text{ and}\\
        & \text{for all } w_1 , \dots, w_n \in \zone^m, \qquad p_{S, b}(w_1, \dots, w_n) \in
        \begin{cases}
            \left[\frac{2}{3},\frac{4}{3}\right], & \text{if } (g(w_1),\ldots,g(w_n))_S = b, \\
            \left[-\frac{1}{3},\frac{1}{3}\right], & \text{otherwise}.
        \end{cases}
    \end{align}
    
    Then, by replacing each variable $z_{S,b}'$ of $q_\text{robust}$ in \Cref{eq:robust} with the corresponding polynomial $p_{S,b}$ from \Cref{eqn:restr}, we get a polynomial $r := q_{\text{robust}}(\ldots,p_{S,b},\ldots) : \zone^{mn} \to \mathbb{R}$ of degree $O(Tn/\sqrt{k})$ that satisfies the following.
    \begin{align*}
    & |r(w_1, \dots, w_n) - \XOR(g(w_1), \dots, g(w_n))|\le \frac{1}{3} \qquad \forall w_1, \dots, w_n \in \zone^{m}\\
    \implies & |r(w) - (\XOR \circ g)(w)|\le \frac{1}{3} \qquad \forall w \in \zone^{mn}.
    \end{align*}

    Due to \Cref{thm:XOR_composition}, we know that $\adeg(\XOR\circ g) = \Omega(n \cdot \adeg(g))$. Putting everything together, we obtain that
    \[\Omega(n \cdot \adeg(g)) = \adeg(\XOR \circ g) \leq \deg(r) = O(Tn/\sqrt{k}),\]
    and since $\mu(f) = k$, we conclude that $\adeg(f \circ g) = T = \Omega(\sqrt{\mu(f)}\adeg(g))$.\qedhere
\end{proof}

From the above proof, we observe that in porting the result from \cite{Ben-David_Bouland_Garg_Kothari_2018} to the substring-query setting, the subcube stifling number shows up naturally in the construction. This motivates the study of this complexity measure, and we establish some of its properties in the upcoming sections.

\section{Subcube stifling number of some Boolean functions}

In this section, we compute the subcube stifling number for several classes of Boolean functions. We start with random functions.

\begin{claim}\label{claim:subcube_stifling_of_random_func}
    Let $f : \zone^n \to \zone$ be a Boolean function chosen uniformly at random. Then with high probability,
    \begin{equation*}
        \mu(f) \in \{\lceil\log n\rceil-1,\lfloor\log n\rfloor\}.
    \end{equation*}
\end{claim}

\begin{proof}

    Fix $S\subseteq [m]$, $|S| \le k, b\in\{0,1\}^S, z \in \zone^{\bar{S}}$. Then, since fixing the variables in $\bar{S}$ to $z$ yields a (uniformly random) Boolean function on $|S|$ variables, we have $\Pr\left[ f_{S\vert z} = \mathbb{I}[x = b] \right] = \frac{1}{2^{2^{|S|}}}$.

    Using the fact $1 + x \leq  e^x$ for all real $x$, we get
    \begin{align*}
        \Pr\left[\forall z\in\{0,1\}^{\bar{S}}, f_{S\vert z} \neq \mathbb{I}[x = b] \right] =  \left(1-\frac{1}{2^{2^{|S|}}}\right)^{2^{n - |S|}} \leq \exp\left(-2^{n - |S| - 2^{|S|}}\right).
    \end{align*}
    
    By a union bound over $S \subseteq [n]$ and $b \in \{0,1\}^S$, we obtain that
    \begin{align*}
    \Pr[\mu(f)<k] &\le \Pr\left[\exists S\subseteq[n],\ |S|\le k,\ \exists b\in\{0,1\}^S \text{ such that } \forall z\in\{0,1\}^{\bar S},\ f_{S\vert z}\neq \mathbb{I}[x=b]\right] \\
    & \le \sum_{\substack{S\subseteq[n]\\ |S|\le k}}\sum_{b\in\{0,1\}^S}\Pr\left[\forall z\in\{0,1\}^{\bar S},\ f_{S\vert z}\neq \mathbb{I}[x=b]\right], \\
    & \le \sum_{s=0}^k \binom{n}{s} 2^{s} \exp\left(-2^{n-s-2^s}\right).
    \end{align*}
    
    Set $k = \log n - 1$. This implies $2^s \leq 2^k \leq n/2$, and hence $n - s - 2^s \geq n - \log n - n/2 = n/2 - \log n$. Plugging this into the above and using $\binom{n}{i} \leq n^i$, we get
    \begin{align*}
    \Pr[\mu(f)<k]
    &\le \exp\left(-\frac{2^{n/2}}{n}\right)\sum_{s=0}^k \binom ns 2^s \le \exp\left(-\frac{2^{n/2}}{n}\right)\sum_{s=0}^k (2n)^s, \\
    &\le \exp\left(-\frac{2^{n/2}}{n}\right)(k+1)(2n)^k\le \exp\left(\ln[(\log n+1)(2n)^{\log n}]-\frac{2^{n/2}}{n}\right) = o(1).
    \end{align*}
    Thus, with high probability, $\mu(f)\ge \log n-1$.

    For the upper bound on $\mu$, let $S \subseteq [n]$, $b \in \{0,1\}^S$ and $z \in \{0,1\}^{\bar{S}}$. Next, let $X_{S,b,z}$ be the random variable that determines whether $f_{S \vert z} = \mathbb{I}[x_S = b]$. Observe from before that $X_{S,b,z} \sim \mathrm{Bernoulli}(1/2^{2^{|S|}})$. For all $k \in [n]$, we define the random variable
    \[N_k := \sum_{\substack{S \subseteq [n] \\ |S| = k}} \sum_{b \in \{0,1\}^S} \sum_{z \in \{0,1\}^{\bar{S}}} X_{S,b,z},\]
    and we observe that
    \[\mathbb{E}\left[N_k\right] = \binom{n}{k} \cdot 2^k \cdot 2^{n-k} \cdot \frac{1}{2^{2^{|S|}}} = \binom{n}{k} 2^{n-2^k}.\]
    Note that $N_k$ counts the expected number of indicator functions of size $k$ in $f$. At the same time, if $\mu(f) \geq k$, then we must have at least $\binom{n}{k}2^k$ such indicator functions in $f$. As such, we obtain by Markov's inequality that
    \[\Pr[\mu(f) \geq k] \leq \Pr\left[N_k \geq \binom{n}{k}2^k\right] \leq \frac{\binom{n}{k}2^{n-2^k}}{\binom{n}{k}2^k} = 2^{n-k-2^k}.\]
    Thus, if we choose $k = \log n$, then we observe that the right-hand side becomes inverse polynomial in $n$, and so $\mu(f) \leq \log n$ with high probability.
\end{proof}

Next, we consider several subclasses of Boolean functions.

\begin{claim}
    Let $f:\{0,1\}^n\rightarrow\{0,1\}$ be a Boolean function.
    \begin{enumerate}
        \item If $f$ is symmetric, then $\mu(f) \le 1$.
        \item If $f$ is monotone, then $\mu(f) = 0$.
    \end{enumerate}
\end{claim}

\begin{proof}
We first show the bound for symmetric $f$, and then for monotone $f$.
    \begin{enumerate}
        \item Suppose $\mu(f) \geq 2$, then we dervie a contradiction. Since $f$ is symmetric, we know that there exists $g:\{0,1,\ldots,n\}\rightarrow \{0,1\}$ such that $f(x) = g(\mathsf{wt}_H(x))$. Let $S = \zone$, $b = 01$. 
        By the definition of subcube stifling number, there exists $c_b\in\{0,1\}^{\bar{S}}$ such that $f_{S\vert c_b}(y) = \mathbb{I}[y = b]$.
        This implies $f_{S\vert c_b}(01) = 1$, and $f_{S\vert c_b}(10) = 0$. However, $f_{S\vert c_b}(01) = g(1+\mathsf{wt}_H(c_b)) = f_{S\vert c_b}(10)$, which is a contradiction.
        
        \item Let $f$ be a monotone function and assume $\mu(f) = 1$. Then for all $i \in [n]$,
        \begin{equation*}
            \exists c\in\{0,1\}^{[n]\setminus\{i\}} \text{ such that } f_{\{i\} \vert c} = \mathbb{I}[y = 0] \implies f_{\{i\} \vert c}(0) = 1 \text{ and } f_{\{i\} \vert c}(1) = 0.
        \end{equation*}
        Since $f$ is a monotone function, any restriction of it is also monotone, yielding a contraction. \qedhere
    \end{enumerate}
    \end{proof}
While most basic symmetric functions such as $\mathsf{AND}$, $\mathsf{OR}$ and $\mathsf{Majority}$ can easily be seen to have subcube stifling number 0, it is not hard to show that $\mu(\mathsf{XOR}) = 1$. Indeed, let $i \in [n]$ and $c_0, c_1\in \{0,1\}^{[n]\setminus \cbra{i}}$ such that $\mathsf{wt}_H(c_0)$ is even and $\mathsf{wt}_H(c_1)$ is odd. Then, $f_{S\vert c_0} = \mathbb{I}[y = 1]$ and $f_{S\vert c_1} = \mathbb{I}[y = 0]$.

\section{Linear codes and subcube stifling number}

In this section, we investigate a special family of Boolean functions, consisting of indicator functions on linear codes $S \leq \mathbb{F}_2^n$. It turns out that these functions have a particularly direct characterization of their subcube stifling number, in terms of the distance and the dual distance of the code.

\subsection{Subcube stifling number and distance of codes}

Recall that we identify $\zone^n$ with the vector space $\mathbb{F}_2^n$.
We say that a vector $x$ is \emph{supported} inside $A \subseteq [n]$ if $\cbra{i \in [n] : x_i = 1} \subseteq A$. For $S \subseteq \mathbb{F}_2^n$ and $T \subseteq [n]$, define the projection
\[
\pi_T : S \to \mathbb{F}_2^T, \qquad \pi_T(x) = x_T.
\]
For the rest of this section, let $S \leq \mathbb{F}_2^n$ be a linear code and let $T \subseteq [n]$.

\begin{lemma}[Injectivity]\label{lem:inj}
$\pi_{\bar T}$ is injective iff there is no nonzero $w \in S$ supported inside $T$.
\end{lemma}

\begin{proof}
If $\pi_{\bar T}$ is not injective, then there exist $u \neq v \in S$ with $u_{\bar T} = v_{\bar T}$. Then $w := u+v \in S$ is nonzero and satisfies $w_{\bar T}=0$, hence is supported inside $T$.

Conversely, if such a $w$ exists, then $\pi_{\bar{T}}(w) = 0 = \pi_{\bar{T}}(0^n)$, so $\pi_{\bar T}$ is not injective.
\end{proof}

\begin{corollary}\label{cor:inj}
If $|T| < d(S)$, then $\pi_{\bar T}$ is injective.
Furthermore, there exists $T$ with $|T| \ge d(S)$ for which $\pi_{\bar T}$ is not injective.
\end{corollary}

\begin{proof}
If $|T| < d(S)$, then every nonzero $x \in S$ has $\mathsf{wt}_H(x) \ge d(S) > |T|$, so no such $x$ can be supported inside $T$. By \Cref{lem:inj}, $\pi_{\bar T}$ is injective.

Conversely, let $x \in S \setminus \{0\}$ have $\mathsf{wt}_H(x) = d(S)$ and take $T = \mathsf{supp}(x)$. Then $|T| = d(S)$ and $x$ is supported inside $T$, so by \Cref{lem:inj}, $\pi_{\bar T}$ is not injective.
\end{proof}

\begin{lemma}[Surjectivity]\label{lem:surj}
$\pi_T$ is surjective iff there is no nonzero $y \in S^\perp$ supported inside $T$.
\end{lemma}

\begin{proof}
If $y \in S^\perp$ is supported inside $T$, then $y_i = 0$ for all $i \in \bar{T}$. Since $y \in S^\perp$, we have every $x \in S$ satisfies $\langle x, y\rangle = \langle x_T, y_T \rangle = 0$.
Thus, \[
\pi_T(S) \subseteq \{u \in \mathbb{F}_2^T : \langle u, y_T\rangle = 0\},
\]
which is a proper subspace of $\mathbb{F}_2^T$ since $y_T \neq 0$, hence $\pi_T$ is not surjective.

Conversely, if $\pi_T(S)$ is not surjective, then it is a proper subspace of $\mathbb{F}_2^T$, so there exists nonzero $z \in \mathbb{F}_2^T$ orthogonal to it. Extending $z$ by zeros outside $T$ gives $y \in S^\perp$ supported inside $T$.
\end{proof}

\begin{corollary}\label{cor:surj}
If $|T| < d(S^\perp)$, then $\pi_T$ is surjective.  
Furthermore, there exists $T$ with $|T| = d(S^\perp)$ for which $\pi_T$ is not surjective.
\end{corollary}

\begin{proof}
If $|T| < d(S^\perp)$, then by definition of $d(S^\perp)$, every nonzero
$y \in S^\perp$ has $\mathsf{wt}_H(y) \ge d(S^\perp) > |T|$, so no such $y$ can be supported inside $T$. Hence $\pi_T$ is surjective by \Cref{lem:surj}.

Conversely, let $y \in S^\perp$ such that $\mathsf{wt}_H(y) = d(S^\perp)$, and take $T = \mathsf{supp}(y)$. Then $|T| = d(S^\perp)$ and $y$ is supported inside $T$, so by the lemma $\pi_T$ is not surjective.
\end{proof}

We are now in shape to state our result regarding the equivalence of the subcube stifling number of the indicator function of a linear code, and the distance of the underlying code and its dual.

\begin{theorem}\label{thm:subcube-stifling_upper_bound_for_codes}
Let $S \leq \mathbb{F}_2^n$ be a linear subspace and $f:\{0,1\}^n\rightarrow\{0,1\}$ be a Boolean function defined as $f = \mathbb{I}[x\in S]$. Then
\[
\mu(f) = \min\big(d(S), d(S^\perp)\big) - 1.
\]
\end{theorem}

\begin{proof}
We first show the lower bound, and then the upper bound.

\noindent\textbf{Lower bound.}
Let $k < \min(d(S), d(S^\perp))$. We will now show that for every $|T|\le k$ and every $b \in \zone^T$, there is $c \in \zone^{\bar{T}}$ such that $f(b,c) = 1$ and $f(b',c) = 0$ for all $b' \neq b$. To that end, observe that

\begin{itemize}
\item $|T| < d(S)$ implies $\pi_{\bar T}$ is injective by Corollary~\ref{cor:inj}.
\item $|T| < d(S^\perp)$ implies $\pi_T$ is surjective by Corollary~\ref{cor:surj}.
\end{itemize}

Now, fix $b \in \zone^T$. By surjectivity of $\pi_T$, there exists $s \in S$ with $s_T = b$. Let $c = s_{\bar T}$. If $(u,c) \in S$, then $(u,c)$ and $(b,c)$ agree on $\bar T$, so by injectivity of $\pi_{\bar{T}}$ they must be equal. Hence $u=b$, proving
\[
f_{T\vert c} = \mathbb{I}[x_T = b].
\]

Thus $\mu(f) \ge k$, and since this holds for all $k < \min(d(S),d(S^{\perp}))$, we have $\mu(f) \geq \min(d(S),d(S^{\perp}))-1$.

\noindent\textbf{Upper bound.}
Let $m = \min(d(S), d(S^\perp))$.

\begin{itemize}
    \item First suppose $m=d(S)$. Choose $x \in S$ with $\mathsf{wt}_H(x)=m$, and let $T=\mathsf{supp}(x)$. By definition, $|T| = m$. Then $x_{\bar T}=0^{\bar T}$. Set $b:=x_T\in\mathbb F_2^T$. Note that $b\neq 0^T$.
    
    We claim that there is no $c\in\mathbb F_2^{\bar T}$ such that $f_{T\vert c}=\mathbb I[u=b]$. Indeed, fix any $c\in\mathbb F_2^{\bar T}$. If $(b,c)\notin S$, then $f_{T\vert c}(b)=0$, so $f_{T\vert c}$ cannot be the indicator of $b$. On the other hand, if $(b,c)\in S$, then since $x=(b,0^{\bar T})\in S$ and $S$ is linear, we have $(b,c)+x=(0^T,c)\in S$.
    Thus both $b$ and $0^T$ are accepted by $f_{T\vert c}$. Since $b\neq 0^T$, this again shows that $f_{T\vert c}$ is not the indicator of $b$. Therefore, for this choice of $T$ and $b$, no restriction outside $T$ realizes the indicator $\mathbb I[u=b]$.
    
    \item Now suppose $m=d(S^\perp)$. By Corollary~\ref{cor:surj}, there exists a set $T\subseteq[n]$ with $|T|=m$ such that the projection $\pi_T:S\to\mathbb F_2^T$ is not surjective. Hence there exists some $b\in\mathbb F_2^T$ such that $b\notin \pi_T(S)$.

    We claim that there is no $c\in\mathbb F_2^{\bar T}$ such that $f_{T\vert c}=\mathbb I[u=b]$. Indeed, fix any $c\in\mathbb F_2^{\bar T}$. Since $b\notin \pi_T(S)$, there is no codeword in $S$ whose restriction to $T$ equals $b$. In particular, $(b,c)\notin S$. Therefore $f_{T\vert c}(b)=0$, so $f_{T\vert c}$ cannot be the indicator of $b$. Thus, for this choice of $T$ and $b$, no restriction outside $T$ realizes the indicator $\mathbb I[u=b]$.
\end{itemize}

Thus $\mu(f) < m$. Combining both bounds gives
\[
\mu(f) = \min(d(S), d(S^\perp)) - 1.\qedhere
\]
\end{proof}

Combining \Cref{thm: dim random} and \Cref{thm:subcube-stifling_upper_bound_for_codes}, we immediately obtain the following corollary, showing that there exist functions with linear subcube stifling number.
\begin{corollary}\label{cor:mulinearlb}
    Let $f$ be the indicator function of a linear code $S$ with $d(S) = \Omega(n)$ and $d(S^\perp) = \Omega(n)$ (such an $S$ is guaranteed to exist by \Cref{thm: dim random}). Then $\mu(f) = \Omega(n)$.
\end{corollary}

\subsection{Indicators of linear codes have large approximate degree}

We now show that indicators of (affine) linear codes have approximate degree at least the dual distance. The proof uses the dual polynomial method (\Cref{lem:dual}). Recall again that we identify the vector space $\mathbb{F}_2^n$ with $\zone^n$.

\begin{lemma}\label{lemma:approx_deg_of_linear_codes}
Let $S \leq \mathbb{F}_2^n$ and $f(x) = \mathbb{I}[x \in S]$. If $S \neq \mathbb{F}_2^n$, then for every $\varepsilon < 1/2$,
\begin{equation*}
    \widetilde{\deg}_{\varepsilon}(f) \geq d(S^\perp).
\end{equation*}
\end{lemma}

\begin{proof}
    We prove the claim by exhibiting a dual witness satisfying the conditions in \Cref{lem:dual}. Let $S\le \mathbb{F}_2^n$ be a linear code. Let $\rho:=\Pr_x[x\in S] = |S|/2^n$. The dual witness is defined as follows.
    \begin{equation*}
        \psi(x) := \frac{1}{2(1-\rho)}\left[ \frac{f(x)}{|S|} - \frac{1}{2^n} \right].
    \end{equation*}

    Now we show that $\psi$, as defined above, satisfies the L1 norm, correlation, and pure high degree conditions.
    \begin{itemize}
        \item \textbf{L1 norm:}
    \begin{equation*}
        \sum_{x\in \mathbb{F}_2^n}|\psi(x)| = \frac{1}{2(1-\rho)}\left[\sum_{x\in S}\left|\frac{1}{|S|} - \frac{1}{2^n}\right| + \sum_{x\not\in S}\left|\frac{1}{2^n}\right|\right] = \frac{1}{2(1-\rho)}\left[ 1-\frac{|S|}{2^n} + \frac{2^n - |S|}{2^n} \right] = 1.
    \end{equation*}

    \item \textbf{Correlation:}
    \begin{equation*}
        \sum_{x\in \mathbb{F}_2^n}\psi(x)f(x) = \sum_{x\in S}\psi(x) = \frac{1}{2(1-\rho)}\sum_{x\in S}\left[ \frac{1}{|S|} - \frac{1}{2^n} \right] = \frac{(1-\rho)}{2(1 - \rho)} = \frac{1}{2} > \varepsilon.
    \end{equation*}

    \item \textbf{Pure high degree:}
    Let $\alpha \in \mathbb{F}_2^n$ such that $\mathsf{wt}_H(\alpha) < d(S^\perp)$. Assume that $\alpha = 0^n$. Then, $\sum_x\psi(x)\chi_\alpha(x) = \sum_x\psi(x) = 0$.
    
    Now assume that $\alpha \neq 0^n$. Since $\mathsf{wt}_H(\alpha) < d(S^\perp)$, we know that $\alpha \not\in S^\perp$ and thus there exists some $y\in S$ such that $\chi_\alpha(y) = (-1)^{\langle x, y\rangle} = -1$. Then, since $S$ is a subspace,
    \begin{equation}\label{eq:zero}
        \sum_{x\in S}\chi_\alpha(x) = \sum_{x\in S}\chi_\alpha(x+y) = \sum_{x\in S}\chi_\alpha(x)\chi_\alpha(y) = -\sum_{x\in S}\chi_\alpha(x) \implies \sum_{x\in S}\chi_\alpha(x) = 0.
    \end{equation}

Using this, for nonzero $\alpha$ with $\mathsf{wt}_H(\alpha) < d(S^\perp)$, we get
\[
\sum_{x\in\mathbb{F}_2^n} \psi(x)\chi_\alpha(x)
=
\frac{1}{2(1-\rho)}
\left(
\frac{1}{|S|}\sum_{x\in S}\chi_\alpha(x)
-
\frac{1}{2^n}\sum_{x \in \mathbb{F}_2^n} \chi_\alpha(x)
\right)
=
\frac{1}{2(1-\rho)|S|}\sum_{x\in S}\chi_\alpha(x)
=
0.
\]
The first equality is substitution of the definition of $\psi$, the second equality uses $\sum_x\chi_\alpha(x)=0$ since $\alpha\neq 0^n$, and the final equality uses \Cref{eq:zero}.
    \end{itemize}

    Therefore, we get that $\widetilde{\deg}_{\varepsilon}(f) \geq d(S^\perp)$ where $\varepsilon < 1/2$.\qedhere
    
\end{proof}

As a corollary, we derive approximate degree lower bounds for indicators of (affine) subspaces.

\begin{corollary}\label{cor:mu_lower_bounds_approx_deg}
    Let $S\le \mathbb{F}_2^n$ and $a\in \mathbb{F}_2^n$. Let $f = \mathbb{I}[x \in a + S]$ be a Boolean function. Then,
    \begin{equation*}
        \adeg(f) = \Omega(\mu(f)).
    \end{equation*}
\end{corollary}

\begin{proof}
    Let $f, g$ be Boolean functions defined as $f = \mathbb{I}[x \in S]$ and $g = \mathbb{I}[x \in a + S]$ where $S\le \mathbb{F}_2^n$ and $a\in \mathbb{F}_2^n$. First we show that $\mu(f) = \mu(g)$.

    Let $T\subseteq [n]$, $|T| \le k$ and $b \in \mathbb{F}_2^T$. We now show that there exists a setting to variables in $\bar{T}$ such that $g_{T \vert c_b'} = \mathbb{I}[x=b]$. This shows $\mu(g) \geq \mu(f)$. We know that there exists $c_b \in \mathbb{F}_2^{\bar{T}}$ such that $f_{T\vert c_b} = \mathbb{I}[x = b + a_T]$. 
    \begin{align*}
        &\mathbb{I}[x = b + a_T] = f_{T\vert c_b}(x) = f((x, c_b)) = g((x, c_b) + a),\\
        & \implies g((x + a_T, c_b + a_{\bar{T}})) = \mathbb{I}[x = b + a_T],\\
        & \implies g((x, c_b + a_{\bar{T}})) = g_{T\vert c_b+a_{\bar{T}}}(x) = \mathbb{I}[x = b].
    \end{align*}
    Since $T\subseteq [n]$, $|T| \le k$ and $b \in \mathbb{F}_2^T$ were chosen arbitrarily, we conclude $\mu(g) \geq \mu(f)$. Essentially the same argument shows that $\mu(f) \geq \mu(g)$. Combining this with \Cref{thm:subcube-stifling_upper_bound_for_codes} and \Cref{lemma:approx_deg_of_linear_codes}, we get that $\adeg(f) = \Omega(\mu(f))$ when Boolean function $f$ is the indicator of a subspace.
    
    To complete the proof, we need to show that $\adeg(f) = \adeg(g)$ for Boolean functions $f = \mathbb{I}[x \in R]$ and $g = \mathbb{I}[x\in a + R]$ where $R\subseteq \mathbb{F}_2^n$, $a\in\mathbb{F}_2^n$. Towards this, we first show that $\adeg(g) \le \adeg(f)$. Let $p$ be a polynomial that $1/3$-approximates $f$ and $\deg(p) = \adeg(f)$. Consider the polynomial $p' : \zone^n \to \mathbb{R}$ defined as follows (below, $+$ refers to addition over reals).
    \begin{equation*}
        p'(x_1,\ldots,x_n) = p(x_1+a_1-2x_1a_1,\ldots,x_n+a_n-2x_na_n).
    \end{equation*}
    This is clearly a degree-preserving transformation, and thus $\deg(p') \leq \deg(p)$. Moreover we have $p'(x) = p(x+a)$ (here $+$ is addition over $\mathbb{F}_2^n$) for all $x \in \zone^n$.
    Hence $p'$ is a $1/3$-approximation to $g$, which implies $\adeg(g) \leq \adeg(f)$. By similar arguments, we get $\adeg(f) \le \adeg(g)$.\qedhere
\end{proof}

Recall \Cref{conj}, which conjectures that $\adeg(f) = \Omega(\mu(f))$ for all Boolean functions $f$. Towards resolving the conjecture in the positive, one attempt might be to show the stronger bound $\lambda(f) = \Omega(\mu(f))$, which would immediately imply the conjecture by \Cref{thm:spectral_sensitivity_lower_bound_adeg}.
We show that this approach is not feasible by showing the existence of Boolean functions for which $\lambda(f) = O(\sqrt{n})$ and $\mu(f) = \Theta(n)$.

\begin{claim}\label{claim:spectral_sensitivity_of_linear_code}
    Let $S \leq \mathbb{F}_2^n$ be a linear code with distance at least $3$. Let $f:\{0,1\}^n\rightarrow\{0,1\}$ be the indicator function of $S$ i.e., $f(x) = \mathbb{I}[x\in S]$. Then, $\lambda(f) = \sqrt{n}$. 
\end{claim}

\begin{proof}
    Since $d(S) \geq 3$, we observe that for any pair of distinct $x,y \in S$, $\mathsf{wt}_H(x \oplus y) \geq 3$. As such, if we have $x \in S$ and $z \in \mathbb{F}_2^n$ such that $\mathsf{wt}_H(x \oplus z) = 1$, then for all $y \in S \setminus \{x\}$, by the reverse triangle inequality
    \[\mathsf{wt}_H(y \oplus z) \geq |\mathsf{wt}_H(y \oplus x) - \mathsf{wt}_H(x \oplus z)| \geq 3 - 1 = 2.\]
    This means that if $z \in \mathbb{F}_2^n$ is Hamming-distance $1$ away from a code word $x \in S$, then $x$ is the unique code word that is this close to $z$.\footnote{Intuitively, this implies that the subgraph induced on the Hamming cube where we only retain the edges that connect two vertices with different function values, falls apart into disjoint star graphs. This is the core observation that allows us to prove that the operator norm of the adjacency matrix is $\sqrt{n}$.}

    Now, let $\mathbf{v} \in \R^{\mathbb{F}_2^n}$. Recall that $A_f[x,y] = 1$ if and only if $f(x) \neq f(y)$ and $\mathsf{wt}_H(x \oplus y) = 1$. Thus, we observe that
    \[\mathbf{v}^TA_f\mathbf{v} = 2\sum_{\substack{(x,y) \in f^{-1}(0) \times f^{-1}(1) \\ \mathsf{wt}_H(x \oplus y) = 1}} v_xv_y = \sum_{\substack{(x,y) \in f^{-1}(0) \times f^{-1}(1) \\ \mathsf{wt}_H(x \oplus y) = 1}} 2(n^{1/4}v_x)\frac{v_y}{n^{1/4}} \leq \sum_{\substack{(x,y) \in f^{-1}(0) \times f^{-1}(1) \\ \mathsf{wt}_H(x \oplus y) = 1}} \sqrt{n}v_x^2 + \frac{v_y^2}{\sqrt{n}},\]
    where we used $2ab \leq a^2 + b^2$ for all $a,b \in \R$. Now, observe that every $y \in f^{-1}(1)$, i.e., $y \in S$, appears $n$ times in the above summation, since all its immediate neighbors, i.e., all bit strings $z \in \mathbb{F}_2^n$ such that $\mathsf{wt}_H(y \oplus z) = 1$, satisfy $z \not\in S$ and so $f(z) = 0$. Similarly, every $x \in f^{-1}(0)$ appears at most once in the summation, as we argued before that there is at most a unique $y \in S$ that is Hamming distance $1$ away from $x$. Thus, we obtain that
    \[
    \mathbf{v}^TA_f\mathbf{v} \leq \sqrt{n} \sum_{x \in f^{-1}(0)} v_x^2 + \frac{n}{\sqrt{n}} \sum_{y \in f^{-1}(1)} v_y^2 = \sqrt{n} \norm{\mathbf{v}}^2,
    \]
    and since this holds for all $\mathbf{v} \in \R^{\mathbb{F}_2^n}$, we obtain that $\lambda(f) = \norm{A_f} \leq \sqrt{n}$.

    It remains to show the lower bound. To that end, let $\mathbf{v} = \mathbf{e}_{0^n} + \sum_{j=1}^n \mathbf{e}_{e_j}/\sqrt{n}$, where $\mathbf{e}_x \in \R^{\mathbb{F}_2^n}$ is the standard basis vector labeled by $x \in \mathbb{F}_2^n$, and $e_j \in \mathbb{F}_2^n$ is the bit string that is $0$ everywhere except for the $j$th position. Now, observe that $\norm{\mathbf{v}}^2 = 2$, and so
    \[\lambda(f) = \norm{A_f} \geq \frac{\mathbf{v}^TA_f\mathbf{v}}{\norm{\mathbf{v}}^2} = \frac12 \cdot 2\sum_{\substack{(x,y) \in f^{-1}(0) \times f^{-1}(1) \\ \mathsf{wt}_H(x \oplus y) = 1}} v_xv_y = \sum_{j=1}^n v_{e_j}v_{0^n} = \sum_{j=1}^n \frac{1}{\sqrt{n}} = \sqrt{n}.\qedhere\]

\end{proof}

We conclude the values of several of the complexity measures for these types of indicator functions on linear codes.

\begin{corollary}\label{cor:mutightness}
    Let $S \leq \mathbb{F}_2^n$ be as in the statement of \Cref{thm: dim random}, and let $f : \zone^n \to \zone$ be defined as $f(x) = \mathbb{I}[x \in S]$. Then,
    \[
    \mathsf{s}(f) = O(\mu(f)), \qquad \deg(f) = O(\mu(f)), \qquad \lambda(f) = O(\sqrt{\mu(f)}).
    \]
\end{corollary}
\begin{proof}
    Since $d(S) = \Omega(n)$ and $d(S^\perp) = \Omega(n)$, \Cref{cor:mulinearlb} implies $\mu(f) = \Omega(n)$. Clearly $s(f) \leq n = O(\mu(f))$ and $\deg(f) \leq n = O(\mu(f))$. Since $d(S) = \Omega(n) > 3$, \Cref{claim:spectral_sensitivity_of_linear_code} implies $\lambda(f) \leq \sqrt{n}$, which proves the last claimed inequality.\qedhere
\end{proof}

\subsection{Connections to binary covering arrays}

We begin this section with the following observation. The support of any Boolean function forms a binary covering array of strength $\mu(f)$.

\begin{definition}[Binary covering arrays]
    Let $N$, $k$ and $t$ be integers satisfying $N\ge 1$ and $1\le k\le t$. Then a \emph{binary covering array} of strength $k$ is a $N\times t$ $0/1$-matrix, denoted by $\mathsf{CA}(N; k,t,2)$, that has each element of $\{0,1\}^k$ appears at least once as a row in any $N\times k$ submatrix.
\end{definition}

See \cite{Lawrence_kacker_Lei_Kuhn_Forbes_2011} for a detailed treatment of binary covering arrays.

\begin{claim}\label{claim:supp_is_CA}
    Let $f:\{0,1\}^n\rightarrow \{0,1\}$ be a Boolean function with $\mu(f) = k$. Let $M$ be a matrix whose rows are all elements in $\mathsf{supp}(f)$. Then, the matrix $M$ is a $\mathsf{CA}(|\mathsf{supp}(f)|; k, n, 2)$ binary covering array.
\end{claim}

\begin{proof}
    The rows of matrix $M$ are indexed by all binary strings in $\mathsf{supp}(f)$ and the columns are indexed by elements in $[n]$. The $(x,i)$ entry of $M$ is the $i$'th bit of the binary string $x\in f^{-1}(1)$. Let $S\subseteq [n]$, $|S| \le k$ and $M_S$ denote the submatrix obtained by selecting the columns whose index is in $S$. Then,
    \begin{equation*}
        \mathsf{rows}(M_S) = \bigcup\limits_{c\in\{0,1\}^{\bar{S}}}\mathsf{supp}(f_{S\vert c}).
    \end{equation*}
    Since, $\mu(f) = k$, we know that, $\forall b\in \{0,1\}^{S}$ $\exists c_b\in\{0,1\}^{\bar{S}}$ such that $\mathsf{supp}(f_{S\vert c_b}) = \{b\}$. This implies that $\mathsf{rows}(M_S) = \{0,1\}^S$ and therefore the matrix $M$ is a $|\mathsf{supp}(f)|\times n$-binary covering array of strength $|S|$.\qedhere
\end{proof}

\begin{definition}[Binary orthogonal array]
    Let $N$, $k$ and $t$ be integers satisfying $N\ge 1$ and $1\le k\le t$. Then a binary orthogonal array of strength $k$ is a $N\times t$ $0/1$-matrix, denoted by $\mathsf{OA}_\lambda(N; k,t,2)$, that has each of the entries in $\{0,1\}^t$ appear exactly $\lambda$ times in the rows of any $N\times k$ submatrix.
\end{definition}
Next, we observe that the support of the indicator function on a linear subspace forms a binary orthogonal array of strength $\mu(f)$. We prove this by combining \Cref{claim:supp_is_CA} with \Cref{claim:linear_supp_is_OA}.

The following claim appears in a discussion in \cite[Pg~304]{Hedayat_Sloane_Stufken_1999}. We provide a short proof for the sake of completeness.

\begin{claim}\label{claim:linear_supp_is_OA}
    Let $M$ be a $\mathsf{CA}(|R|;k,n,2)$ binary covering array such that $x,y\in \mathsf{rows}(M) \implies x\oplus y \in \mathsf{rows}(M)$. Then $M$ is also a $\mathsf{OA}_\lambda(|R|;k,n,2)$ binary orthogonal array.
\end{claim}

\begin{proof}
    Let $R$ be the set of bit strings formed by the rows in the binary covering array. Let $S\subseteq [n]$, $|S| = t$. Let $\lambda$ be the multiplicity of $0^S$ in $R$, that is, the number of rows in the table that are all-zeros when restricted to $S$. Let $x\in R$ such that $x_S \neq 0^t$ (such a $x$ is guaranteed to exists since the rows form a binary covering array of strength $t$). Define the set $R_x$ as follows.

    \begin{equation*}
        R_x := \{ x\oplus y \mid y\in R \} \subseteq R, \qquad (\text{by linearity}).
    \end{equation*}

    Note that the multiplicity of $x_S$ in $R_x$ is $\lambda$ and hence its multiplicity in $R$ is at least $\lambda$. If the multiplicity in $R$ is greater than $\lambda$, then $0^S$ has the same multiplicity in $R_x$. This is a contradiction and therefore, multiplicity of $x_S$ in $R$ is exactly $\lambda$.

    Observing that the above arguments hold for any arbitrary $x_S\in\{0,1\}^t\setminus\{0^t\}$ completes the proof.\qedhere
    
\end{proof}

\begin{theorem}[{\cite[Theorem~4.6]{Hedayat_Sloane_Stufken_1999}}]\label{thm:linear_code_equals_OA}
    If $S$ is a $[n,t,d]_2$ binary linear code with dual distance $d^{\perp}$, then the codewords of $C$ form the rows of a $\mathsf{OA}_\lambda(n;d^{\perp}-1,t,2)$ binary orthogonal array. Conversely, the rows of a linear $\mathsf{OA}_\lambda(n;k,t,2)$ binary orthogonal array form a $[n,t,d]_2$ binary linear code with dual distance $d^{\perp} = k+1$.
\end{theorem}

We note that the original statement of \Cref{thm:linear_code_equals_OA} shows the equivalence of linear codes and orthogonal arrays over $\mathbb{F}_q$. We choose to present it for $\mathbb{F}_2$ since that is all we require. Combining \Cref{thm:linear_code_equals_OA} with \Cref{lemma:approx_deg_of_linear_codes} gives us the following corollary.

\begin{corollary}\label{cor:adeg_lower_bound_OA}
    Let $M$ be a $\mathsf{OA}_\lambda(N;k,n,2)$ be a binary orthogonal array. Let $f:\{0,1\}^n\rightarrow\{0,1\}$ be an indicator function of rows in the matrix $M$. Then,
    \begin{equation*}
        \adeg(f) = \Omega(\mu(f)).
    \end{equation*}
\end{corollary}

\section*{Acknowledgments}

The authors gratefully acknowledge Subhasree Patro for helpful discussions, and OpenAI’s ChatGPT (GPT-5.4 and GPT-5.5), which assisted in developing several crucial ideas presented in this work. The authors bear full responsibility for any mistakes in the work. NR thanks Benjamin Jany and Mattia Cipro for helpful discussions. AC is supported by a Simons-CIQC postdoctoral fellowship through NSF QLCI Grant No.~2016245 and NR acknowledges the support from the Dutch Ministry of Education, Culture, and Science through Gravitation project ``Challenges in Cyber Security - 024.006.037” for this work.

\bibliography{references}

\appendix

\section{Proof of \Cref{thm: dim random}}

Below is a standard fact about sums of binomial coefficients.
\begin{fact}\label{fact: binom}
    Let $H_2(x) = -x\log x - (1-x)\log(1-x)$ denote the binary entropy function. For all positive integers $n$, $k \leq n/2$ we have
    \[
    \sum_{i = 0}^k \binom{n}{i} \leq 2^{H_2(k/n)\cdot n}.
    \]
\end{fact}

\begin{proof}[Proof~of~\Cref{thm: dim random}]
For simplicity we assume $n$ to be even; the whole argument can be easily made to work with $n$ odd by using ceilings and floors appropriately. Let $k= n/2$, and choose $S\le \mathbb{F}_2^n$ uniformly at random among all $k$-dimensional subspaces. We first bound the probability that $S$ contains a nonzero vector of small Hamming weight. Fix $x\in\zone^n\setminus\cbra{0}$. By symmetry,
\[
\Pr[x\in S]=\frac{2^k-1}{2^n-1}\le 2^{k-n}.
\]
Therefore,
\begin{align*}
\Pr[d(S)<\delta n] & \leq \sum_{\substack{x\in\zone^n\setminus\{0\}\\|x|<\delta n}} \Pr[x\in S], \tag*{by a union bound}\\
& \leq \left(\sum_{i = 0}^{\delta n}\binom{n}{i}\right)2^{k-n},\\
& \le 2^{H_2(0.11)\cdot n}\cdot 2^{-n/2}, \tag*{by Fact~\ref{fact: binom}}\\
& < 1/2,
\end{align*}
where the last inequality follows since $n$ is sufficiently large and $H_2(0.11) \approx 0.4999 < 1/2$.

Since $S$ was a random $n/2$-dimensional subspace of $\zone^n$, this means $S^\perp$ is also a random $n/2$-dimensional subspace of $\zone^n$. Using the same argument, we conclude that $\Pr[d(S^\perp)<\delta n] < 1/2$.
Finally, another union bound gives
\[
\Pr\left[d(S)<\delta n\ \text{or}\ d(S^\perp)<\delta n\right] < 1.
\]
Thus, for all sufficiently large $n$, there exists a $n/2$-dimensional subspace $S\leq\zone^n$ such that
\[
d(S)\geq \delta n \qquad\text{and}\qquad d(S^\perp)\geq \delta n.
\]
This completes the proof.\qedhere
\end{proof}

\end{document}